\documentclass{article}

\usepackage{PRIMEarxiv}

\usepackage[utf8]{inputenc} 
\usepackage[T1]{fontenc}    
\usepackage{hyperref}       
\usepackage{url}            
\usepackage{booktabs}       
\usepackage{amsfonts}       
\usepackage{amsmath,amssymb}
\usepackage{mathtools}
\usepackage{nicefrac}       
\usepackage{microtype}      
\usepackage{lipsum}
\usepackage{fancyhdr}       
\usepackage{graphicx}       

\usepackage{times}
\usepackage{subfig}
\usepackage{xcolor}
\usepackage{array}
\usepackage{cite}
\usepackage{multirow}
\usepackage{adjustbox}
\usepackage[scr=dutchcal]{mathalfa}
\usepackage[font=small,labelfont=bf]{caption}
\usepackage{enumitem}
\usepackage{tikz}
\usetikzlibrary{arrows,matrix} 

\DeclareMathOperator{\diag}{diag}

\DeclareMathOperator{\rank}{rank}

\newtheorem{theorem}{Theorem}

\newtheorem{proposition}{Proposition}
\newtheorem{corollary}{Corollary}

\newtheorem{definition}{Definition}

\newtheorem{remark}{Remark}

\newtheorem{proof}{Proof}

\usepackage{tikz-network}
\tikzset{
    block/.style = {draw, fill=white, rectangle, minimum height=.75cm, minimum width=1.25cm},
    Circle/.style = {draw, thick, fill=white, circle, minimum height=.75cm, minimum width=1.25cm},
    tmp/.style  = {coordinate}, 
    sum/.style= {draw=white, fill=white, circle, minimum height=.75cm, minimum width=1.25cm},
    input/.style = {coordinate},
    output/.style= {coordinate},
    pinstyle/.style = {pin edge={to-,thin,black}
    }
}

\title{
Cooperative $\mathcal{H}_\infty$ Fault-Tolerant Tracking with \\ ISS Guarantees for Networked Systems \\
with Sensor Faults
}

\author{
  Moh Kamalul Wafi, Bambang L. Widjiantoro, Katherin Indriawati\\
  Department of Engineering Physics \\
  Institut Teknologi Sepuluh Nopember (ITS), Surabaya, Indonesia\\
  \texttt{\{kamalul.wafi\}@its.ac.id} \\
   \And
  Yurid E. Nugraha \\
  Department of Electrical Engineering \\
  Institut Teknologi Sepuluh Nopember (ITS), Surabaya, Indonesia\\
}

\begin{document}
\maketitle

\begin{abstract}
This paper develops a cooperative fault-tolerant tracking framework for heterogeneous networked linear systems subject to sensor faults and external disturbances. Each unit employs an augmented $\mathcal{H}_\infty$ observer that jointly reconstructs the system state and unknown sensor fault, providing disturbance-attenuated estimation guarantees. An inner state-feedback gain is synthesized through convex $\mathcal{H}_\infty$ Linear Matrix Inequalities (LMIs) to ensure robust closed-loop stabilization and disturbance rejection, while an outer distributed integral action eliminates steady-state tracking offsets and enables cooperative tracking of a setpoint source. 
The resulting cooperative error dynamics are shown to satisfy an Input-to-State Stability (ISS) property with respect to disturbances and residual estimation uncertainty, and converge exponentially to zero in the disturbance-free case. Furthermore, vanishing cooperative error guarantees network-wide consensus tracking of the desired setpoint. Numerical studies on heterogeneous DC-motor networks with star, cyclic, and path communication topologies demonstrate accurate state and fault estimation, robust cooperative tracking, and resilience against disturbances and time-varying sensor faults. The proposed framework provides a scalable and robust coordination strategy for interconnected systems operating under sensing imperfections and uncertain environments.
\end{abstract}
\allowdisplaybreaks


\section{Introduction}

Fault--tolerant control (FTC) has been extensively investigated for safety--critical single-agent systems, including industrial processes, aerospace vehicles, and autonomous platforms \cite{R1,R2,R3}. As modern infrastructures increasingly rely on interconnected subsystems, distributed architectures have become essential due to their scalability, robustness to individual failures, and reduced communication requirements \cite{R4}. A key challenge in such systems is maintaining reliable state estimation under limited and potentially corrupted measurements. Consequently, distributed estimation has received significant attention for both linear and nonlinear multi-agent systems \cite{Wafi-AIP,R5,Wafi-DistEstTank,Wafi-Elham}.

In networked environments, sensor faults are particularly disruptive because corrupted measurements propagate through cooperative control loops and may destabilize the entire network even if only a single unit is affected. This has motivated the development of distributed FTC schemes that address sensor faults via observer-based compensation \cite{R8,R9}, actuator faults via resilient control laws \cite{R10,R11}, and combined fault scenarios through adaptive or reconfigurable methods \cite{R12,R13}. Although these contributions demonstrate meaningful progress, most existing methods assume nominal or disturbance-free plant dynamics. A second limitation is that current distributed FTC approaches rarely establish Input--to--State Stability (ISS) guarantees for the cooperative tracking error. ISS guarantees are crucial for quantifying the robustness of cooperative tracking against persistent disturbances, modeling uncertainties, and residual estimation errors. Such effects inevitably arise in large-scale interconnected systems and can significantly degrade network performance if not properly addressed.

Motivated by these gaps, this paper develops a cooperative FTC architecture that combines augmented fault estimation, robust distributed control, and distributed integral feedback to achieve reliable reference tracking across heterogeneous agents. In contrast to many existing approaches, the proposed framework provides $\mathcal{H}_\infty$--based robustness conditions formulated as convex Linear Matrix Inequalities (LMIs) \cite{R14}, together with a network-level ISS guarantee with respect to disturbances and estimation imperfections. Furthermore, the proposed strategy ensures exact consensus tracking when faults and disturbances vanish. These properties make the framework suitable for applications such as UAV formations, vehicle platoons, distributed sensor networks with drifting biases, and industrial robotic teams.

The main contributions of this paper are:
\begin{enumerate}[leftmargin=*]
    \item \textbf{Augmented observer design.}
    We develop an augmented observer that jointly estimates the system state and unknown sensor faults for each networked unit. The $\mathcal{H}_\infty$ condition of Theorem~\ref{thm:observer-Hinf} guarantees robust estimation performance in the presence of disturbances and sensor degradation.

    \item \textbf{Robust cooperative fault-tolerant control architecture.}
    Using the estimated states, we construct a cooperative fault-tolerant control strategy composed of an inner $\mathcal{H}_\infty$ state-feedback loop and an outer distributed integral tracking loop. The resulting controller stabilizes the interconnected network while attenuating disturbances and compensating steady-state tracking offsets.

    \item \textbf{$\mathcal{H}_\infty$ synthesis via convex LMIs.}
    The observer and controller gains are synthesized through convex Linear Matrix Inequality (LMI) conditions derived in Theorems~\ref{thm:observer-Hinf} and~\ref{thm:network-Hinf}, enabling systematic robust design for heterogeneous networked systems.

    \item \textbf{ISS-based cooperative tracking analysis.}
    Through Theorem~\ref{thm:ISS-error} together with Proposition~\ref{prop:e-zero-implies-consensus}, we establish Input-to-State Stability (ISS) of the cooperative tracking error with respect to disturbances and residual estimation uncertainty. In the disturbance-free case, the cooperative tracking error converges exponentially to zero.

    \item \textbf{Validation under heterogeneous topologies and time-varying sensor faults.}
    Numerical studies on heterogeneous multi-agent networks with star, cyclic, and path communication topologies demonstrate accurate fault estimation, robust cooperative tracking, and resilience against disturbances and time-varying sensor faults.
\end{enumerate}

\noindent\textbf{Notation.}
For a positive integer $p$, $\mathbf{1}_p \in \mathbb{R}^p$ and $\mathbf{0}_p \in \mathbb{R}^p$ denote the vectors of all ones and all zeros, respectively, while $I_p$ denotes the $p\times p$ identity matrix. Moreover, $0_{p\times n}$ denotes the zero matrix of size $p\times n$.
For a matrix $M$, $M^\top$ denotes its transpose, and $\diag\{\cdot\}$ denotes a block diagonal matrix constructed from its arguments. 
The Kronecker product is denoted by $\otimes$. For matrices $A \in \mathbb{R}^{n\times m}$ and $B \in \mathbb{R}^{p\times q}$, the Kronecker product $A \otimes B \in \mathbb{R}^{np\times mq}$ is defined in the standard way. 

For a collection of vectors $x_i \in \mathbb{R}^n$, $i=1,\dots,m$, we define the stacked vector $\bar{x} = [x_1^\top,\dots,x_m^\top]^\top \in \mathbb{R}^{mn}$. 
Throughout the paper, $\|\cdot\|$ denotes the Euclidean norm, and $\mathcal{L}_2[0,\infty)$ denotes the space of square-integrable signals.

\section{Communication Network}\label{sec:ComNetwork}

Information transfer among the agents is described by a weighted directed graph (digraph) 
$\mathcal{G}=(\mathcal{V},\mathcal{E},\mathcal{W})$, where $\mathcal{V}$ is the set of nodes (agents), 
$\mathcal{E} \subseteq \mathcal{V}\times\mathcal{V}$ is the set of directed edges representing communication links, 
and $\mathcal{W} = [w_{ij}]$ is the nonnegative weight matrix, where $w_{ij}$ quantifies the information received by agent $i$ from agent $j$. 
In particular, $w_{ij} > 0$ if and only if $(i,j)\in\mathcal{E}$.

The network under study comprises $m{+}1$ agents collected in the set 
$\mathcal{V}=\{0,1,\dots,m\}$, where agent~$0$ acts as a setpoint source delivering reference or supervisory information, while agents $1$ to $m$ represent interconnected units (e.g., physical subsystems, autonomous devices, cyber--physical modules, or distributed industrial components) equipped with local sensing, estimation, and fault-mitigation capabilities.

The in-neighborhood of agent $i$ is defined as $\mathcal{N}_i=\{\,j\in\mathcal{V}\mid(i,j)\in\mathcal{E}\,\}$. To distinguish unit--unit communication from source broadcasting, the graph $\mathcal{G}$ is decomposed into two induced subgraphs: 
\begin{enumerate}
    \item The first, $\mathcal{G}_m=(\mathcal{V}_m,\mathcal{E}_m,\mathcal{W}_m)$ with $\mathcal{V}_m=\{1,\dots,m\}$, captures inter-unit interactions, where $\mathcal{W}_m$ denotes the corresponding edge-weight set inherited from $\mathcal{W}$. For this subgraph $\mathcal{G}_m$, the adjacency and in-degree matrices are $[\mathbb{A}_m]_{ij}=w_{ij}$ and $\mathbb{D}_m=\diag\{d_1,\dots,d_m\}$ where $d_i=\sum_{j:(i,j)\in\mathcal{E}_m}w_{ij}$, leading to the Laplacian matrix $\mathbb{L}_m = \mathbb{D}_m - \mathbb{A}_m$. This Laplacian encodes the relative information exchange among units.
    
    \item The second, $\mathcal{G}_0=(\mathcal{V}_0,\mathcal{E}_0,\mathcal{W}_0)$ with $\mathcal{V}_0=\{0\}\cup\{i:(i,0)\in\mathcal{E}\}$, captures source-to-unit connectivity, where $\mathcal{W}_0$ is defined analogously.
    For the subgraph $\mathcal{G}_0$, broadcasts from the source are represented by the matrix $\mathbb{A}_0=\diag\{w_{10},\dots,w_{m0}\}$. 
\end{enumerate}

An example of this decomposition is shown in Fig.~\ref{Fig:network}. The full communication structure is captured by the augmented (pinned) Laplacian $\mathbb{L} = \mathbb{L}_m + \mathbb{A}_0$. 
Define $w_i = d_i + w_{i0}$ as the total incoming weight to agent $i$, and let $\mathbb{W}=\diag\{w_1,\dots,w_m\}$. When $w_i=1$ for all $i$, i.e., $\mathbb{W}=I_m$, this gives
\begin{equation*}
    \mathbb{L}\coloneqq \mathbb{L}_m + \mathbb{A}_0 = \mathbb{W}-\mathbb{A}_m.
\end{equation*}
Finally, the distributed communication satisfies the balance condition
\begin{equation}\label{eq:ComNet:balance}
    (\mathbb{L}-\mathbb{A}_0)\mathbf{1}_m=\mathbf{0}_m
    \quad\Leftrightarrow\quad
    (\mathbb{A}_m+\mathbb{A}_0)\mathbf{1}_m=\mathbf{1}_m,
\end{equation}
meaning each unit receives the same total incoming weight from its neighbors and from the setpoint source.  
When $\mathbb{W}\neq I_m$, normalized weights 
$\tilde w_{ij}=w_{ij}/w_i$ and $\tilde w_{i0}=w_{i0}/w_i$ yield the normalized (row-stochastic) Laplacian 
$\tilde{\mathbb{L}}=\tilde{\mathbb{L}}_m+\tilde{\mathbb{A}}_0$, which also satisfies \eqref{eq:ComNet:balance}.

\begin{remark}\label{rem:threshold}
If at least one unit directly receives the setpoint signal ($w_{i0}>0$) and every unit is reachable from node~$0$ along a directed communication pathway contained in $\mathcal{E}$, then the augmented Laplacian $\mathbb{L}$ is positive stable. 
Thus, the supervisory reference affects the entire industrial network.
\end{remark}

\section{Problem Formulation}\label{sec:ProbFor}

We study a network consisting of a setpoint source denoted as agent~$0$ and $m$ interconnected units, each representing a generic dynamical subsystem whose physical output must track the setpoint in the presence of disturbances and sensor degradation. 
The dynamics of the $i$-th unit with disturbance and sensor degradation can be written as
\begin{equation}\label{eq:ProbFor:state-space}
    \begin{aligned}
    \dot{x}_i(t) &= A_ix_i(t) + B_iu_i(t) + D_i v_{i}(t), \\
    y_{f,i}(t) &= C_ix_i(t) + F_if_{s,i}(t).
    \end{aligned}
\end{equation}
Here, $x_i(t)\in\mathbb{R}^{n_x}$ is the local state and $u_i(t)\in\mathbb{R}^{n_u}$ is the control input to be determined later. The signal $y_{f,i}(t)\in\mathbb{R}^{n_y}$ denotes the measured output in the presence of the sensor
degradation $f_{s,i}(t)\in\mathbb{R}^{n_y}$. The pair $(A_i,B_i)$ is controllable and $(A_i,C_i)$ is observable, with the matrices of appropriate dimensions.
The disturbance matrix is $D_i\in\mathbb{R}^{n_x\times n_v}$ with disturbance signal $v_i(t)\in\mathbb{R}^{n_v}$ which belongs to $\mathcal{L}_2[0,\infty)$. Here, the fault-location matrix is assumed to be known, with $F_i \coloneqq I_{n_y}$. To reflect real-world diversity, the model \eqref{eq:ProbFor:state-space} is allowed to be heterogeneous across the network, meaning that the system matrices $(A_i,B_i,C_i,D_i)$ may differ for each unit.

The setpoint source broadcasts its reference output $y_0$ to the units that receive information from node~$0$, while
each unit $i$ shares its estimated output $\hat{y}_i$ with its neighbors. For every unit $i$, the total in-neighbor setpoint $z_i$ and the corresponding local tracking error $e_i$ are defined as
\begin{equation}\label{eq:ProbFor:z_and_e}
    z_i = \sum\nolimits_{j\in\mathcal{N}_{i}} w_{ij}\hat y_j + w_{i0} y_0, \qquad e_i = \hat{y}_i - z_i.
\end{equation}
The signal $z_i$ represents a weighted combination of information received by agent $i$ from its neighbors and, if available, from the setpoint source. In particular, $\hat{y}_j$ denotes the estimated output shared by neighboring agents $j$, while $y_0$ is the source output directly available only to agents with $w_{i0} > 0$.
This networked structure enables distributed estimation and tracking using only local communication. 

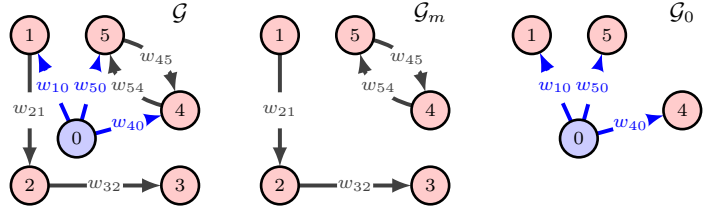
\begin{figure}[t!]
    \centering
    \scalebox{1}{{\begin{tikzpicture}
            \centering
            \Text[x=1,y=1.3,fontsize=\small]{$\mathcal{G}$};
            \Vertex[x=-.375,y=-.375,label=$0$,color=blue,opacity=0.2,size=.5]{L}
            \Vertex[x=-1,y=1,label=$1$,color=red,opacity=0.2,size=.5]{1}
            \Vertex[x=-1,y=-1,label=$2$,color=red,opacity=0.2,size=.5]{2}
            \Vertex[x=1,y=-1,label=$3$,color=red,opacity=0.2,size=.5]{3}
            \Vertex[x=1,y=0,label=$4$,color=red,opacity=0.2,size=.5]{4}
            \Vertex[x=0,y=1,label=$5$,color=red,opacity=0.2,size=.5]{5}
            \Edge[Direct,color=blue,label=$w_{10}$](L)(1)
            \Edge[Direct,color=blue,label=$w_{40}$](L)(4)
            \Edge[Direct,color=blue,label=$w_{50}$](L)(5)
            \Edge[Direct,label=$w_{21}$](1)(2)
            \Edge[Direct,label=$w_{32}$](2)(3)
            \Edge[Direct,bend=30,label=$w_{45}$](5)(4)
            \Edge[Direct,bend=30,label=$w_{54}$](4)(5)
        \end{tikzpicture}}}
    \qquad
    \scalebox{1}{{\begin{tikzpicture}
            \centering
            \Text[x=1,y=1.3,fontsize=\small]{$\mathcal{G}_m$};
            \Vertex[x=-1,y=1,label=$1$,color=red,opacity=0.2,size=.5]{1}
            \Vertex[x=-1,y=-1,label=$2$,color=red,opacity=0.2,size=.5]{2}
            \Vertex[x=1,y=-1,label=$3$,color=red,opacity=0.2,size=.5]{3}
            \Vertex[x=1,y=0,label=$4$,color=red,opacity=0.2,size=.5]{4}
            \Vertex[x=0,y=1,label=$5$,color=red,opacity=0.2,size=.5]{5}
            \Edge[Direct,label=$w_{21}$](1)(2)
            \Edge[Direct,label=$w_{32}$](2)(3)
            \Edge[Direct,bend=30,label=$w_{45}$](5)(4)
            \Edge[Direct,bend=30,label=$w_{54}$](4)(5)
        \end{tikzpicture}}}
    \qquad
    \scalebox{1}{{\begin{tikzpicture}
            \centering
            \Text[x=1,y=1.3,fontsize=\small]{$\mathcal{G}_0$};
            \Vertex[x=-.375,y=-.375,label=$0$,color=blue,opacity=0.2,size=.5]{L}
            \Vertex[x=-1,y=-1,style={fill=white,opacity=0},size=.5]{2}
            \Vertex[x=-1,y=1,label=$1$,color=red,opacity=0.2,size=.5]{1}
            \Vertex[x=1,y=0,label=$4$,color=red,opacity=0.2,size=.5]{4}
            \Vertex[x=0,y=1,label=$5$,color=red,opacity=0.2,size=.5]{5}
            \Edge[Direct,color=blue,label=$w_{10}$](L)(1)
            \Edge[Direct,color=blue,label=$w_{40}$](L)(4)
            \Edge[Direct,color=blue,label=$w_{50}$](L)(5)
        \end{tikzpicture}}}
    \caption{Example of a graph $\mathcal{G}$ with $m=5$, unit-to-unit subgraph $\mathcal{G}_m$, and source-to-unit subgraph $\mathcal{G}_0$, showing the decoupling and assignment of $w_{ij}$.}
    \label{Fig:network}
\end{figure}

To characterize robustness against external disturbances and quantify disturbance attenuation in the proposed cooperative framework, we adopt the standard $\mathcal{H}_\infty$ performance criterion given below.

\begin{definition}
A system with bounded disturbance $v(t)$ is said to achieve $\mathcal{H}_\infty$ performance if the following hold:
\begin{enumerate}
    \item The system is asymptotically stable when $v(t)\equiv 0$.
    \item With zero
    initial condition, there exists a constant $\gamma>0$ such that for any disturbance $v(t)\in\mathcal{L}_2[0,\infty)$,
    \begin{equation*}
        \int_{0}^{\infty} x^{\top}(t)x(t)\,dt
        \le \gamma^{2} \int_{0}^{\infty} v^{\top}(t)v(t)\,dt,
    \end{equation*}
    where $x(t)$ is the state vector of the system.
\end{enumerate}
\end{definition}

Using the $\mathcal{H}_\infty$ notion introduced above, then the proposed framework pursues three main objectives:

\begin{enumerate}[leftmargin=*]

\item \textbf{Observer design.}
Develop an augmented observer that provides reliable state and fault estimates under disturbances. The corresponding $\mathcal{H}_\infty$ performance is enforced through the LMI condition established in Theorem~\ref{thm:observer-Hinf}.

\item \textbf{Robust distributed control.} With point (1), 
design an inner feedback gain $\mathbf{K}$, using the LMI in Theorem~\ref{thm:network-Hinf} ($\mathcal{H}_\infty$ bound), together with an outer distributed integral action that enables setpoint tracking propagated through communication network while maintaining robustness.

\item \textbf{Consensus tracking.}
Ensure that all units converge to the setpoint source and to one another. Theorem~\ref{thm:ISS-error} with Proposition~\ref{prop:e-zero-implies-consensus} guarantees Input–to–State Stability (ISS) of the cooperative tracking error across the entire network.

\end{enumerate}

\section{Cooperative Fault--Tolerant Tracking}

This section embeds the local unit dynamics of Section~\ref{sec:ProbFor} into the
communication network of Section~\ref{sec:ComNetwork} and develops a
cooperative fault--tolerant control (FTC) strategy for compensating sensor
degradation while achieving distributed setpoint tracking. For brevity, the
time argument $(t)$ is omitted when no confusion arises.

Node~$0$ provides the source, or setpoint, signal $y_0\in\mathbb{R}^{n_y}$. At the network level, we define
$\bar y_0 = \mathbf{1}_m\otimes y_0$. Moreover, the stacked dynamics of the
$m$ units are given by
\begin{equation}\label{eq:FTC:network-dynamics}
    \dot{\bar x} = \mathbf{A}\bar x + \mathbf{B}\bar u + \mathbf{D}\bar v,  \qquad
    \bar y_f = \mathbf{C}\bar x + \mathbf{F}\bar f_s.
\end{equation}
To simplify notation, we introduce stacked vectors and block-diagonal matrices.
For any signal $g_i \in \mathbb{R}^{n}$, $i=1,\dots,m$, define the stacked vector
$\bar g = [g_1^\top,\dots,g_m^\top]^\top \in \mathbb{R}^{mn}$.
Accordingly, we define $\bar x$, $\bar u$, $\bar v$, $\bar y_f$, and $\bar f_s$.
For matrices, let $G_i \in \mathbb{R}^{n\times n}$ and define the block-diagonal matrix
$\mathbf{G} = \diag\{G_1,\dots,G_m\} \in \mathbb{R}^{mn\times mn}$.
In particular, $\mathbf{A}$, $\mathbf{B}$, $\mathbf{C}$, $\mathbf{D}$, and $\mathbf{F}$ are constructed in this manner.

The stacked in-neighbor setpoint is defined as
\begin{equation}\label{eq:FTC:z}
    \bar z = [z_1^\top,\dots,z_m^\top]^\top
    :=
    (\mathbb{A}_m\otimes I_{n_y})\hat{\bar y}
    + (\mathbb{A}_0\otimes I_{n_y})\bar y_0,
\end{equation}
which compactly represents the local aggregation rule in \eqref{eq:ProbFor:z_and_e} for all units.
Using $\mathbb{W} = I_m$ and \eqref{eq:ComNet:balance}, the cooperative tracking error is defined as
\begin{equation}\label{eq:FTC:coop-error}
    \bar e
    :=
    (\mathbb{L}\otimes I_{n_y})\hat{\bar y}
    - (\mathbb{A}_0\otimes I_{n_y})\bar y_0
    = \hat{\bar y} - \bar z.
\end{equation}
The error $\bar e$ captures the mismatch between each unit’s estimated output and its aggregated reference signal. 
When $\bar e = 0$, the estimated outputs of all units achieve a weighted consensus around the source $y_0$.

Before proceeding to controller and observer design, it is essential to understand the structural meaning of the cooperative tracking error~\eqref{eq:FTC:coop-error}. The following result establishes that this error vanishes exactly when all agents reach agreement with the setpoint satisfying the conditions in Remark~\ref{rem:threshold}.

\begin{proposition}\label{prop:e-zero-implies-consensus}
    Assume that Remark~\ref{rem:threshold} holds such that the augmented Laplacian $\mathbb{L}$ is positive stable (i.e., all eigenvalues have positive real parts) and the weights $w_{ij}, w_{i0}\ge 0$ satisfy $(\mathbb{A}_m+\mathbb{A}_0)\mathbf{1}_m=\mathbf{1}_m$. 
    Therefore $\bar e = 0$ if and only if $\hat y_i = y_0$ for all $i\in\{1,\dots,m\}$.
\end{proposition}

Thus, enforcing $\bar e = 0$ guarantees network-wide consensus and tracking of the source $y_0$. The remainder of this subsection develops an augmented observer capable of estimating both the system state and the unknown sensor fault despite disturbances and sensor degradation.

\subsection{Cooperative Observer Design}

Since accurate tracking requires reliable state information, unknown sensor faults render direct use of the measured outputs $\bar y_f$ inadequate, motivating the need for an augmented state–fault observer.

To estimate both the system state and the sensor fault, we introduce an augmented state representation. Define the augmented state as $\bar x_a = [\bar x^\top, \bar f_s^\top]^\top$.
The corresponding augmented networked system is given by
\begin{equation} \label{eq:FTC:combined-dynamics}
    \begin{aligned}
    \mathbf{E}_1\dot{\bar x}_a &= \mathbf{A}_a\bar x_a + \mathbf{B}\bar u + \mathbf{D}\bar v, \\
    \bar y_f &= \mathbf{E}_2\bar x_a,
    \end{aligned}
\end{equation}
where $\mathbf{A}_a = [\mathbf{A},\; 0_{\bar n_x\times \bar n_y}]$,
$\mathbf{E}_1 = [I_{\bar n_x},\; 0_{\bar n_x\times \bar n_y}]$, and 
$\mathbf{E}_2 = [\mathbf{C},\; \mathbf{F}]$ with $\mathbf{F} = I_{\bar n_y}$.
More specifically, $\mathbf{E}_1$ selects the state component of the augmented vector, while $\mathbf{E}_2$ maps the augmented state to the measured output.
Since $\rank([\mathbf{E}_1^\top,\mathbf{E}_2^\top]^\top)=\bar n_x+\bar n_y$, 
the matrix $[\mathbf{E}_1^\top,\mathbf{E}_2^\top]^\top$ is nonsingular. 
Define
\begin{equation*}
    \mathbf{F}_1 \coloneqq [I_{\bar n_x},\;-\mathbf{C}^\top]^\top,
    \qquad
    \mathbf{F}_2 \coloneqq [0_{\bar n_x\times \bar n_y}^\top,\;I_{\bar n_y}]^\top .
\end{equation*}
These matrices provide an explicit inverse transformation between the descriptor-like augmented representation and the standard augmented state coordinates. Then,
\begin{equation*}
    \begin{bmatrix}
        \mathbf{E}_1 \\
        \mathbf{E}_2 
    \end{bmatrix}
    \begin{bmatrix}
        \mathbf{F}_1 & \mathbf{F}_2
    \end{bmatrix}
    =
    \begin{bmatrix}
        \mathbf{F}_1 & \mathbf{F}_2
    \end{bmatrix}
    \begin{bmatrix}
        \mathbf{E}_1 \\
        \mathbf{E}_2 
    \end{bmatrix}
    = I_{\bar n_x + \bar n_y},
\end{equation*}
which implies $([\mathbf{E}_1^\top, \mathbf{E}_2^\top]^\top)^{-1} = [\mathbf{F}_1, \mathbf{F}_2]$.

Multiplying $\mathbf{F}_1$ to both sides of \eqref{eq:FTC:combined-dynamics} and using the identity $\mathbf{F}_1\mathbf{E}_1 + \mathbf{F}_2\mathbf{E}_2 = I_{\bar n_x + \bar n_y}$, we obtain
\begin{equation} \label{eq:FTC:augmented-dynamics}
    \begin{aligned}
    \mathbf{F}_1 \mathbf{E}_1 \dot{\bar x}_a 
    &= \mathbf{F}_1\mathbf{A}_a \bar x_a + \mathbf{F}_1\mathbf{B}\bar u + \mathbf{F}_1\mathbf{D} \bar v, \\
    \dot{\bar x}_a 
    &= \mathbf{F}_1\mathbf{A}_a \bar x_a + \mathbf{F}_1\mathbf{B}\bar u + \mathbf{F}_1\mathbf{D} \bar v + \mathbf{F}_2\mathbf{E}_2 \dot{\bar x}_a,
    \end{aligned}
\end{equation}
where the term $\mathbf{F}_2\mathbf{E}_2 \dot{\bar x}_a$ arises from reconstructing the full augmented dynamics using the inverse transformation.

Consider the virtual observer
\begin{equation} \label{eq:FTC:virtual-observer}
    \begin{aligned}
    \dot{\bar x}_o
    &= \mathbf{F}_1\mathbf{A}_a \bar x_o + \mathbf{F}_1\mathbf{B}\bar u + \mathbf{F}_2\mathbf{E}_2 \dot{\bar x}_a 
    + \mathbf{L}( \bar y_f - \mathbf{E}_2 \bar x_o),
    \end{aligned}
\end{equation}
where $\bar x_o$ denotes the observer state (i.e., an estimate of the augmented state $\bar x_a$), and $\mathbf{L}$ is a gain matrix to be designed.
Define the estimation error $\bar \epsilon = \bar x_a - \bar x_o$. Subtracting \eqref{eq:FTC:augmented-dynamics} from \eqref{eq:FTC:virtual-observer}, the estimation error dynamics become
\begin{equation}\label{eq:FTC:estimation-error}
    \dot{\bar \epsilon} = (\mathbf{F}_1\mathbf{A}_a - \mathbf{L}\mathbf{E}_2)\bar \epsilon + \mathbf{F}_1\mathbf{D} \bar v.
\end{equation}

With the augmented representation in place, the first objective is to design an observer that guarantees robust estimation under sensor degradation. This is achieved by the following $\mathcal{H}_\infty$ condition.

\begin{theorem} \label{thm:observer-Hinf}
    For a given constant $\delta>0$, the estimation error dynamics \eqref{eq:FTC:estimation-error} is asymptotically stable with disturbance attenuation level $\delta$ if there exist a matrix $\mathbf{P}\succ 0 \in\mathbb{R}^{(\bar n_x+\bar n_y)\times(\bar n_x+\bar n_y)}$ and a matrix $\mathbf{H}\in\mathbb{R}^{(\bar n_x+\bar n_y)\times \bar n_y}$ such that the following LMI holds:
    \begin{equation} \label{eq:FTC:LMI-observer}
        \Pi =
        \begin{bmatrix}
        \Delta & \mathbf{P}\mathbf{F}_1\mathbf{D} \\
        *      & -\delta^2 I_{\bar n_v}
        \end{bmatrix} < 0,
    \end{equation}
    where $\Delta = \mathbf{P}\mathbf{F}_1\mathbf{A}_a + \mathbf{A}_a^{\top}\mathbf{F}_1^{\top}\mathbf{P} - \mathbf{H}\mathbf{E}_2 - \mathbf{E}_2^{\top}\mathbf{H}^{\top} + I_{\bar n_x + \bar n_y}.$
\end{theorem}

This condition ensures that the observer provides a disturbance-attenuated estimate of the augmented state, guaranteeing that the effect of disturbances on the estimation error is bounded in the $\mathcal{H}_\infty$ sense.

According to Definition~1, the estimation error dynamics \eqref{eq:FTC:estimation-error} using the virtual observer \eqref{eq:FTC:virtual-observer} satisfies $\mathcal{H}_{\infty}$ performance with disturbance attenuation level $\delta$. However, since the term $\mathbf{F}_2\mathbf{E}_2 \dot{\bar x}_a$ depends on the unknown augmented state, the virtual observer \eqref{eq:FTC:virtual-observer} is not directly implementable. To eliminate this dependence, we introduce a state transformation.
Define the new observer state $\bar \eta = \bar x_o - \mathbf{F}_2\mathbf{E}_2 \bar x_a$ and differentiating gives
\begin{equation} \label{eq:FTC:observer}
    \begin{cases}
    \dot{\bar \eta}
    &= \mathbf{F}_1\mathbf{A}_a \bar x_o + \mathbf{F}_1\mathbf{B}\bar u + \mathbf{L}( \bar y_f - \mathbf{E}_2 \bar x_o), \\
    &= (\mathbf{F}_1\mathbf{A}_a - \mathbf{L}\mathbf{E}_2) \bar \eta + \mathbf{F}_1\mathbf{B}\bar u + \big[(\mathbf{F}_1\mathbf{A}_a - \mathbf{L}\mathbf{E}_2)\mathbf{F}_2 + \mathbf{L}\big] \bar y_f, \\
    \bar x_o &= \bar \eta + \mathbf{F}_2 \bar y_f,
    \end{cases}
\end{equation}
where $\mathbf{L} := \mathbf{P}^{-1}\mathbf{H}$. 
This transformation eliminates the dependence on the unknown term $\dot{\bar x}_a$, yielding an implementable observer driven solely by measurable signals.
Theorem~\ref{thm:observer-Hinf} guarantees that each unit obtains a disturbance–attenuated estimate of its augmented state. This estimate is now used to synthesize an $\mathcal{H}_\infty$ state–feedback law that ensures closed–loop stability in the presence of both unit disturbances and observer-induced perturbations.

\subsection{Robust Cooperative Control}

Using the observer estimate, consider the control law
\begin{equation} \label{eq:FTC:control}
    \bar u = \bar u^\ast - \mathbf{I}(\mathbf{K}_\ell,\bar e),
\end{equation}
where $\mathbf{I}(\cdot)$ denotes a distributed integral (or consensus-based) feedback term defined based on the tracking error.

In Theorem~\ref{thm:network-Hinf}, we design an inner feedback gain $\mathbf{K}$ by setting $\bar u^\ast \coloneqq \mathbf{K}\mathbf{E}_1 \bar x_o$. 
Introducing $\bar \vartheta \coloneqq \mathbf{K}\mathbf{E}_1 \bar \epsilon$, $\bar \theta = [\bar v^\top,\bar \vartheta^\top]^\top$, and $\mathbf{B}_\theta \coloneqq [\mathbf{D}, -\mathbf{B}]$, the closed-loop networked system can be written as
\begin{equation} \label{eq:FTC:closed-loop}
    \begin{aligned}
    \dot{\bar x}
    &= (\mathbf{A} + \mathbf{B}\mathbf{K})\bar x - \mathbf{B}\mathbf{K}\mathbf{E}_1 \bar \epsilon + \mathbf{D}\bar v \\
    &= (\mathbf{A} + \mathbf{B}\mathbf{K})\bar x + \mathbf{B}_\theta \bar \theta,
    \end{aligned}
\end{equation}
where $\bar \epsilon = \bar x_a - \bar x_o$ and $\mathbf{E}_1\bar x_a = \bar x$. 
We now proceed to the second objective: designing a networked control gain that attenuates the effect of $\bar\theta$ on the global state.

\begin{theorem}\label{thm:network-Hinf}
Given the observer gain $\mathbf{L}$ of Theorem~\ref{thm:observer-Hinf}, and let $\alpha>0$ and $\delta>0$, where $\delta$ is the attenuation level from Theorem~\ref{thm:observer-Hinf}.
If there exist matrices
$\mathbf{R}\succ 0\in\mathbb{R}^{\bar n_x\times\bar n_x}$ and
$\mathbf{G}\in\mathbb{R}^{\bar n_u\times\bar n_x}$ such that the LMI
\begin{equation}\label{eq:FTC:LMI-control-4x4}
    \Lambda =
    \begin{bmatrix}
        \bar\Delta & \mathbf{R} & -\mathbf{B} & \mathbf{D} \\
        *          & -I_{\bar n_x}         & 0           & 0                 \\
        *          & *          & -\alpha I_{\bar n_u}   & 0                 \\
        *          & *          & *           & -\delta^{2} I_{\bar n_v}
    \end{bmatrix} < 0,
\end{equation}
holds, where $\bar\Delta \coloneqq \mathbf{A}\mathbf{R} + \mathbf{R}\mathbf{A}^\top + \mathbf{B}\mathbf{G} + \mathbf{G}^\top\mathbf{B}^\top$,
then the inner feedback gain $\mathbf{K} \coloneqq \mathbf{G}\mathbf{R}^{-1}$
renders the closed--loop networked system \eqref{eq:FTC:closed-loop} asymptotically stable and guarantees the $\mathcal{H}_\infty$ performance $\|\bar x\|_{2} < \gamma\|\bar v\|_{2},$
where the attenuation level $\gamma$ is given by
$\gamma = \sqrt{(\alpha \lambda_{\max}(\mathbf{K}^\top \mathbf{K}) + 1)\delta^2}.$
\end{theorem}

\begin{remark}\label{rem:gamma-delta}
In view of Theorem~\ref{thm:observer-Hinf}, the estimation error satisfies an $\mathcal{H}_\infty$ bound with attenuation level $\delta>0$. Since $\bar\vartheta = \mathbf{K}\mathbf{E}_1\bar\epsilon$, it follows that
\begin{equation*}
    \alpha\|\bar\vartheta\|_{2}^2
    \le \alpha\|\mathbf{K}\mathbf{E}_1\|^2\|\bar\epsilon\|_{2}^2
    \le \alpha\|\mathbf{K}\|^2\|\bar\epsilon\|_{2}^2
    \le \alpha\lambda_{\max}(\mathbf{K}^\top\mathbf{K})\,\delta^2\|\bar v\|_{2}^2,
\end{equation*}
Consequently, the weighted input energy satisfies
\begin{equation*}
    \|\bar\theta\|_2^2 = \delta^2\|\bar v\|_2^2 + \alpha\|\bar\vartheta\|_2^2
    \le \big(\alpha \lambda_{\max}(\mathbf{K}^\top \mathbf{K}) + 1\big)\delta^2 \|\bar v\|_2^2,
\end{equation*}
which leads to the closed-loop attenuation level $\gamma$.
\end{remark}

This inner loop $\mathbf{K}$ in Theorem~\ref{thm:network-Hinf} guarantees robust stabilization and disturbance attenuation. However, $\bar u^\ast$ alone does not ensure tracking of the network setpoint. To remove steady-state tracking offsets, we introduce an outer-loop integral state $\bar\xi$ driven by the cooperative tracking error,
\begin{equation*}
    \dot{\bar\xi} = \bar e, \qquad \bar e = \hat{\bar y}-\bar z,
\end{equation*}
and define the tracking term as $\mathbf{I}(\cdot)\coloneqq \mathbf{K}_\ell\bar\xi$. 

The gain $\mathbf{K}_\ell$ is selected as a diagonal positive gain matrix and tuned to regulate the speed of the outer tracking loop. In practice, $\mathbf{K}_\ell$ is chosen after the inner gain $\mathbf{K}$ has been synthesized, so that the integral action evolves on a slower time scale than the stabilized inner-loop dynamics. This separation mitigates excessive overshoot while preserving the disturbance-attenuation properties of the inner loop.

By feeding back this accumulated distributed tracking error, the outer loop compensates steady-state offsets caused by disturbances, heterogeneity, or sensor degradation, and drives the estimated outputs toward the reference $y_0$ according to the communication structure satisfying Remark~\ref{rem:threshold}.

\subsection{ISS Cooperative Tracking}

With both observer and controller satisfying their $\mathcal{H}_\infty$ objectives, we now examine the resulting cooperative error dynamics to complete the third objective. The next theorem shows that the network achieves ISS cooperative tracking.

\begin{theorem}\label{thm:ISS-error}
Assume that \eqref{eq:ComNet:balance} and Remark~\ref{rem:threshold} hold, and
let the gains $\mathbf{L}$ and $\mathbf{K}$ be designed according to
Theorems~\ref{thm:observer-Hinf} and~\ref{thm:network-Hinf}.
Let $\bar x_0$ be a constant reference state corresponding to the setpoint output $\bar y_0$, and consider the closed--loop networked system \eqref{eq:FTC:closed-loop} with
$\bar \vartheta_\ast = \bar \vartheta_1 + \bar \vartheta_2$,
$\bar \vartheta_1 \coloneqq \mathbf{K}\mathbf{E}_1 \bar \epsilon$,
$\bar \vartheta_2 \coloneqq \mathbf{K}_\ell \bar \xi$,
$\bar \theta^\ast = [\bar v^\top,\bar \vartheta_\ast^\top]^\top$,
and $\mathbf{B}_\theta \coloneqq [\mathbf{D},-\mathbf{B}]$.
Define the cooperative state error
\begin{equation}\label{eq:FTC:coop-error-output}
    \bar e_x
    :=
    (\mathbb{L}\otimes I_{n_x})\bar x
    - (\mathbb{A}_0\otimes I_{n_x})\bar x_0 .
\end{equation}
Then there exist positive constants $c_1,c_2,c_3>0$ such that the following
ISS estimate holds for all $t\ge 0$:
\begin{equation}\label{eq:ISS-bound}
    \|\bar e_x(t)\|
    \le
    c_1 e^{-c_2 t}\|\bar e_x(0)\|
    + c_3 \sup_{0\le \tau\le t}\|\bar\theta^\ast(\tau)\|.
\end{equation}
If $\bar\theta^\ast\equiv 0$, then $\bar e_x(t)\to 0$ exponentially.
\end{theorem}

\begin{remark}\label{rem:time-varying-x0}
In Theorem~\ref{thm:ISS-error}, we assumed a constant reference $\bar x_0$, so that $\bar \phi_0$ is constant and can be absorbed into the equilibrium $\bar e^\star$. 
If the reference $\bar x_0$ is time--varying, then differentiating
\eqref{eq:FTC:coop-error-output} gives
\begin{equation*}
    \dot{\bar e}_x
    = \Phi \bar e_x
    + \bar\phi_0
    + \mathbf{B}_\phi \bar\theta^\ast
    - (\mathbb{A}_0\otimes I_{n_x})\dot{\bar x}_0,
\end{equation*}
where
$\bar\phi_0 =
(\mathbb{L}\otimes I_{n_x})(\mathbf{A}+\mathbf{B}\mathbf{K})
(\mathbb{L}\otimes I_{n_x})^{-1}
(\mathbb{A}_0\otimes I_{n_x})\bar x_0 .$
Thus, $\bar e_x$ is driven by two inputs: the disturbance $\bar\theta^\ast$ and the reference rate $\dot{\bar x}_0$. Using the Lyapunov function $V(\bar e_x) = \bar e_x^{\top}\mathbf{P}_e\bar e_x$ with $\mathbf{P}_e\succ 0$, and the same arguments as in the proof of Theorem~\ref{thm:ISS-error}, one obtains an estimate of the form
\begin{equation}\label{eq:ISS-e}
    \dot V
    \le
    -\alpha^\ast \|\bar e_x\|^2
    + \beta_1^\ast \|\bar\theta^\ast\|^2
    + \beta_2^\ast \|\dot{\bar x}_0\|^2,
\end{equation}
for some $\alpha^\ast,\beta_1^\ast,\beta_2^\ast>0$. 
This inequality is an input--to--state stability (ISS) of $\bar e_x$ with respect to the inputs $[\bar\theta^{\ast\top},\dot{\bar x}_0^\top]^\top$.
\end{remark}

Remark~\ref{rem:time-varying-x0} establishes ISS of the cooperative \emph{state}
error under disturbances and time--varying references. Since the regulated
outputs are linear functions of the states, this result can be directly
translated into cooperative tracking guarantees at the output level, which is
the quantity of practical interest in many applications. This is formalized in
the following corollary.

\begin{corollary}\label{cor:output-tracking}
Under the conditions of Theorem~\ref{thm:ISS-error}, define the cooperative
output error
\begin{equation*}
    \bar e_y
    :=
    (\mathbb{L}\otimes I_{n_y})\bar y
    - (\mathbb{A}_0\otimes I_{n_y})\bar y_0,
\end{equation*}
where $\bar y = \mathbf{C}\bar x$ and $\mathbf{C}=\diag\{C_1,\dots,C_m\}$.
Then $\bar e_y$ is ISS with respect to $\bar\theta^\ast$. In particular, there exist constants $\tilde c_1,\tilde c_2,\tilde c_3>0$ such that
\begin{equation*}
    \|\bar e_y(t)\|
    \le
    \tilde c_1 e^{-\tilde c_2 t}\|\bar e_y(0)\|
    + \tilde c_3 \sup_{0\le \tau\le t}\|\bar\theta^\ast(\tau)\|.
\end{equation*}
If $\bar\theta^\ast\equiv 0$, then $\bar e_y(t)\to 0$ exponentially. Consequently, by Proposition~\ref{prop:e-zero-implies-consensus}, the network achieves cooperative output consensus tracking of the setpoint.
\end{corollary}

\section{Numerical Example}
\begin{figure}[h!]
    \centering
    \scalebox{1}{{\begin{tikzpicture}
            \centering
            \Text[x=0,y=1.5]{$\mathcal{G}_{\mathrm{s}}\coloneqq$ Star};
            \Vertex[x=0,y=0,label=$0$,color=blue,opacity=0.2,size=.5]{L}
            \Vertex[x=1,y=1,label=$1$,color=red,opacity=0.2,size=.5]{1}
            \Vertex[x=-1,y=1,label=$2$,color=red,opacity=0.2,size=.5]{2}
            \Vertex[x=-1,y=-1,label=$3$,color=red,opacity=0.2,size=.5]{3}
            \Vertex[x=1,y=-1,label=$4$,color=red,opacity=0.2,size=.5]{4}
            \Edge[Direct,color=blue,label=$1.0$](L)(1)
            \Edge[Direct,color=blue,label=$1.0$](L)(2)
            \Edge[Direct,color=blue,label=$1.0$](L)(3)
            \Edge[Direct,color=blue,label=$1.0$](L)(4)
        \end{tikzpicture}}}
    \qquad
    \scalebox{1}{{\begin{tikzpicture}
            \centering
            \Text[x=0,y=1.5]{$\mathcal{G}_{\mathrm{c}}\coloneqq$ Cyclic};
            \Vertex[x=0,y=0,label=$0$,color=blue,opacity=0.2,size=.5]{L}
            \Vertex[x=1,y=1,label=$1$,color=red,opacity=0.2,size=.5]{1}
            \Vertex[x=-1,y=1,label=$2$,color=red,opacity=0.2,size=.5]{2}
            \Vertex[x=-1,y=-1,label=$3$,color=red,opacity=0.2,size=.5]{3}
            \Vertex[x=1,y=-1,label=$4$,color=red,opacity=0.2,size=.5]{4}
            \Edge[Direct,color=blue,label=$0.4$](L)(1)
            \Edge[Direct,color=blue,label=$0.4$](L)(2)
            \Edge[Direct,color=blue,label=$0.4$](L)(3)
            \Edge[Direct,color=blue,label=$0.4$](L)(4)
            \draw[<->,ultra thick, latex' -latex'] (1) -- node[midway, fill=white, inner sep=2pt]{\small $0.3$} (2);
            \draw[<->,ultra thick, latex' -latex'] (2) -- node[midway, fill=white, inner sep=2pt]{\small $0.3$} (3);
            \draw[<->,ultra thick, latex' -latex'] (3) -- node[midway, fill=white, inner sep=2pt]{\small $0.3$} (4);
            \draw[<->,ultra thick, latex' -latex'] (4) -- node[midway, fill=white, inner sep=2pt]{\small $0.3$} (1);
        \end{tikzpicture}}}
    \qquad
    \scalebox{1}{{\begin{tikzpicture}
            \centering
            \Text[x=0,y=1.5]{$\mathcal{G}_{\mathrm{p}}\coloneqq$ Path};
            \Vertex[x=0,y=0,label=$0$,color=blue,opacity=0.2,size=.5]{L}
            \Vertex[x=1,y=1,label=$1$,color=red,opacity=0.2,size=.5]{1}
            \Vertex[x=-1,y=1,label=$2$,color=red,opacity=0.2,size=.5]{2}
            \Vertex[x=-1,y=-1,label=$3$,color=red,opacity=0.2,size=.5]{3}
            \Vertex[x=1,y=-1,label=$4$,color=red,opacity=0.2,size=.5]{4}
            \Edge[Direct,color=blue,label=$1.0$](L)(1)
            \draw[<->,ultra thick, -latex'] (1) -- node[midway, fill=white, inner sep=2pt]{\small $1.0$} (2);
            \draw[<->,ultra thick, -latex'] (2) -- node[midway, fill=white, inner sep=2pt]{\small $1.0$} (3);
            \draw[<->,ultra thick, -latex'] (3) -- node[midway, fill=white, inner sep=2pt]{\small $1.0$} (4);
        \end{tikzpicture}}}
    \caption{Three network topologies $(\mathcal{G}_{\mathrm{s}}, \mathcal{G}_{\mathrm{c}}, \mathcal{G}_{\mathrm{p}})$ with weights used in the simulations.}
    \label{Fig:topology}
\end{figure}
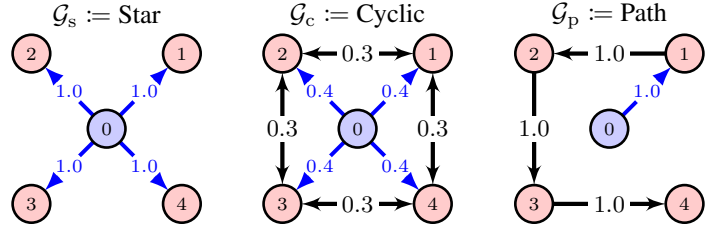

In this section, we illustrate Theorems~\ref{thm:observer-Hinf}--\ref{thm:ISS-error}
on a network of $m=4$ heterogeneous DC--motor units \cite{WAFI-JRC}. For each unit $i$,
the local dynamics follow \eqref{eq:ProbFor:state-space} with
\begin{equation*}
    A_i = \begin{bmatrix}
        -\frac{b_i}{J} & \frac{M_i}{J} \\[2pt]
        -\frac{M_i}{L_i} & -\frac{R_i}{L_i}
    \end{bmatrix}, \quad
    B_i = \begin{bmatrix}
        0 \\[2pt] \frac{1}{L_i}
    \end{bmatrix}, \quad
    D_i = \sigma(i)\begin{bmatrix}
        1 \\[3pt] 1
    \end{bmatrix},
\end{equation*}
and $C_i = [1,\,0]$, where the physical parameters are chosen as $J=0.01$,
$b_i = 0.1(1+0.1(i-1))$, $M_i = 0.01(1+0.05(i-1))$, $R_i = 1.0(1-0.02(i-1))$,
$L_i = 0.5(1+0.03(i-1))$, and $\sigma(i) = 0.1\,i$ for $i=1,\dots,m$.

Each disturbance channel is driven by a constant input $v_i(t)\equiv 0.1$, and a time-varying sensor fault of the form
\begin{equation*}
    f_{s,i}(t) =
    \begin{cases}
        0, & t \le 10\,\mathrm{s}, \\
        3 + 2\sin\big(0.2(t-10)\big), & t > 10\,\mathrm{s},
    \end{cases}
\end{equation*}
is injected at $t=10\,\mathrm{s}$.

The initial unit states are chosen as
$x_i(0)\sim\mathcal{U}([-1,1]^2)$, $i=1,\dots,m$, while the observer states are initialized at $\eta_i(0)=\mathbf{0}_{2}$, yielding $\hat x_i(0)=\mathbf{0}_{2}$.

The observer gain $\mathbf{L}$ is obtained by solving the LMI of
Theorem~\ref{thm:observer-Hinf} with attenuation level $\delta=0.3$.
The inner feedback gain $\mathbf{K}$ is computed from
Theorem~\ref{thm:network-Hinf} using a common design parameter
$\alpha_i \equiv 0.2$.
To ensure setpoint tracking, each agent employs the distributed integral action described in the form
\[
\bar u = \mathbf{K}\mathbf{E}_1 \bar x_o - \mathbf{K}_\ell\,\bar\xi, 
\qquad \dot{\bar\xi} = \bar e,
\]
with $\mathbf{K}_\ell = \diag\{\ell_1,\dots,\ell_4\}$, $\ell_i = 90$.
This outer-loop integral correction eliminates steady-state tracking error and enables convergence to piecewise-constant setpoints across all communication topologies considered.

Figure~\ref{F3a} shows the estimation performance of the proposed observer.
The first two subplots report the state estimation errors $(x_i - \hat x_i)$
for the $m$ agents, corresponding to $x_i^{(1)}$ (motor speed) and
$x_i^{(2)}$ (armature current), respectively. The third subplot depicts the
true and estimated sensor faults.

In all cases, the state estimation errors converge rapidly to zero despite the
persistent disturbances and the time-varying sensor fault activated at
$t=10\,\mathrm{s}$. Moreover, the estimated fault closely tracks the true fault signal,
demonstrating accurate reconstruction even under time-varying conditions.
These results confirm the $\mathcal{H}_\infty$ observer performance guaranteed
by Theorem~\ref{thm:observer-Hinf} and highlight the robustness of the proposed
observer to non-constant sensor degradations.

Figures~\ref{F3b}--\ref{F3d} depict the closed-loop trajectories of $x_i^{(1)}$
for the star, cyclic, and path topologies. Each plot includes the true states
$x_i^{(1)}$, the estimates $\hat x_i^{(1)}$, the piecewise-constant setpoint,
and a vertical dashed line marking the activation of the sensor fault at
$t=10\,\mathrm{s}$. For all three topologies, the agents exhibit a short
transient following each step change and the fault activation, and then
converge to the common setpoint.

The star graph achieves the fastest settling time and smallest overshoot,
while the cyclic and path graphs show slightly slower but still well-damped
responses, in line with their weaker connectivity. The performance degradation
relative to the centralized (star) case is modest, illustrating that the
proposed fault-tolerant control architecture is robust to heterogeneity,
time-varying sensor faults, and disturbances as long as the communication
graph is connected (Remark~\ref{rem:threshold},
Proposition~\ref{prop:e-zero-implies-consensus}) and the LMI conditions of
Theorems~\ref{thm:observer-Hinf}--\ref{thm:network-Hinf} are satisfied.

\begin{figure*}[t!]
    \centering
    \subfloat[\label{F3a} Error]{\includegraphics[width=.24\linewidth]{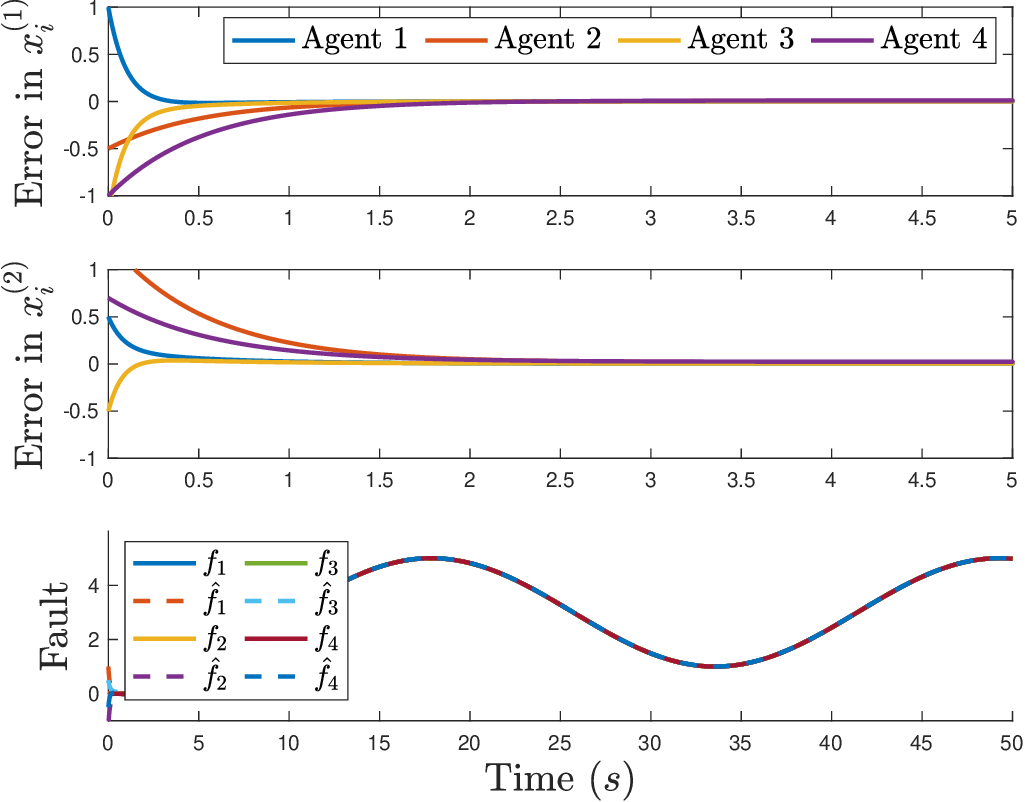}}
    \subfloat[\label{F3b} Star--$\mathcal{G}_{\mathrm{s}}$]{\includegraphics[width=.245\linewidth]{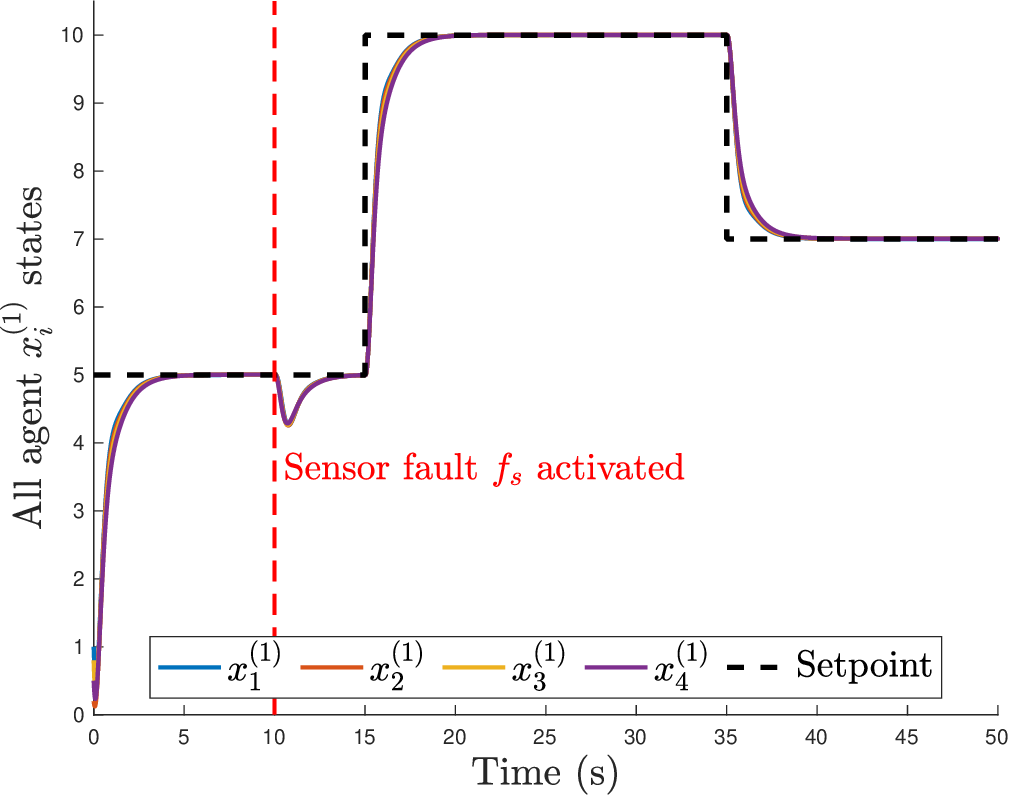}}
    \subfloat[\label{F3c} Cyclic--$\mathcal{G}_{\mathrm{c}}$]{\includegraphics[width=.245\linewidth]{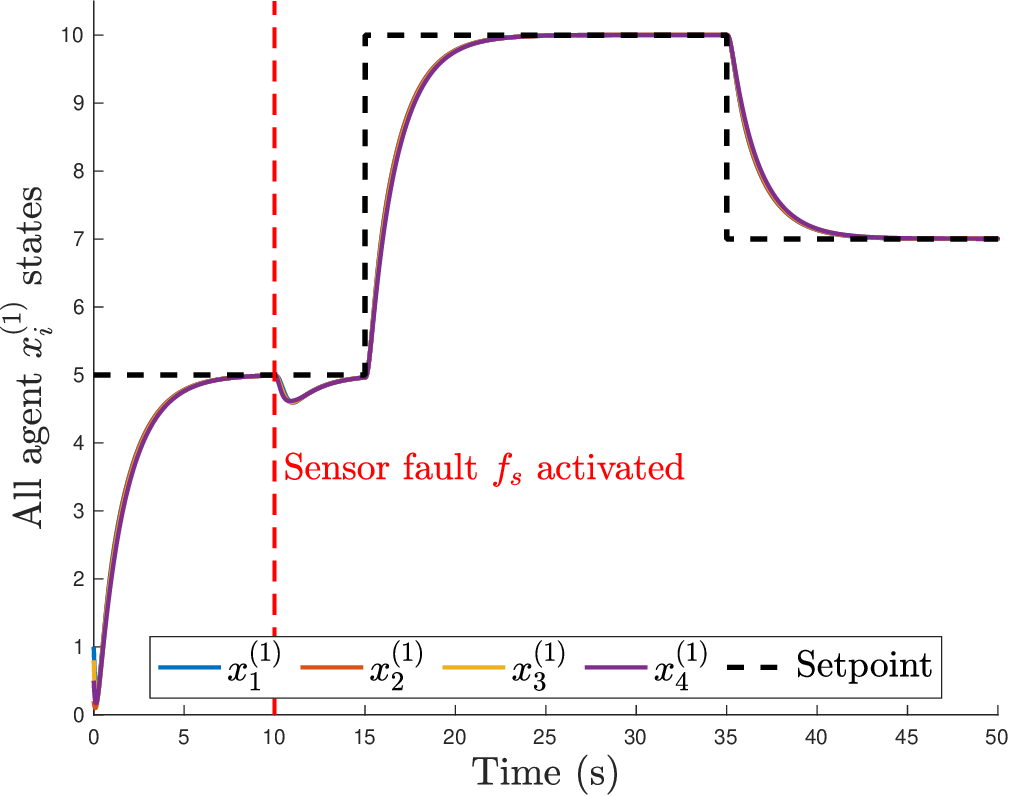}}
    \subfloat[\label{F3d} Path--$\mathcal{G}_{\mathrm{p}}$]{\includegraphics[width=.245\linewidth]{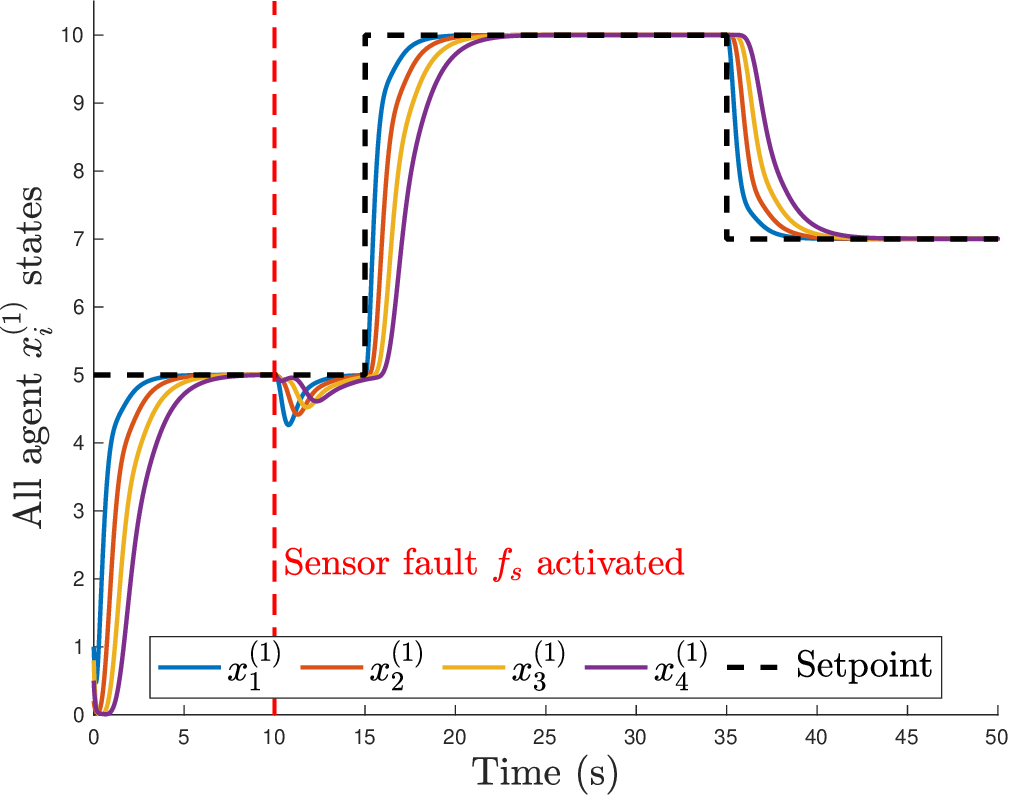}}
    \caption{Trajectories of $m = 4$ sensing nodes for three topologies, along with corresponding error norms.}
    \label{fig:topologies}
\end{figure*}

\section{Conclusion}

This paper presented a cooperative fault-tolerant tracking framework for networked linear systems subject to sensor faults and external disturbances. The proposed framework combines three complementary components:
(i) an augmented observer that reconstructs both system states and unknown sensor faults while guaranteeing an $\mathcal{H}_\infty$ disturbance-attenuation property;
(ii) an inner $\mathcal{H}_\infty$ state-feedback controller synthesized through convex LMIs to ensure robust closed-loop stability and disturbance rejection; and
(iii) an outer distributed integral action that eliminates steady-state tracking offsets and enables cooperative tracking of the network setpoint.

Using an ISS-based analysis, the cooperative tracking error was shown to remain uniformly bounded in the presence of disturbances and residual estimation uncertainty, while converging exponentially to zero in the disturbance-free case. Furthermore, Proposition~1 established that vanishing cooperative error guarantees network-wide consensus tracking of the desired setpoint. Numerical studies on heterogeneous DC-motor networks with star, cyclic, and path communication topologies demonstrated accurate state and fault estimation, robust cooperative tracking, and resilience against disturbances and time-varying sensor faults. Overall, the proposed framework provides a scalable and robust solution for coordinated control of interconnected systems operating under sensing imperfections and uncertain environments.

Future work will investigate adaptive and data-driven extensions, including feedback-dependent tuning of observer and controller gains, as well as learning-based mechanisms for improved fault estimation and uncertainty quantification from data. Additional directions include extending the framework to broader classes of nonlinear and uncertain systems, incorporating communication constraints such as delays and packet losses, and addressing more general adversarial and time-varying attack scenarios. Integration with reinforcement-learning--based decision-making and distributed optimization methods will also be explored to further enhance autonomy and performance in complex networked environments \cite{Add1,Add2,Wafi-SIAM,Add3,R17,Add4,Wafi-ACC24,Wafi-CoDIT,Add5}.

\bibliographystyle{ieeetr}
\bibliography{reference.bib}

@ARTICLE{R1,
  author={Zhang, Xuelin and Xu, Xiaobin and Li, Jianning and Ma, Feng and Zhang, Zhenjie and Brunauer, Georg and Steyskal, Felix},
  journal={IEEE Sensors Journal}, 
  title={Fault Estimation and $\mathcal{H}_{\infty}$ Fuzzy Active Fault-Tolerant Control Design for Ship Steering Autopilot Subject to Actuator and Sensor Faults}, 
  year={2023},
  volume={23},
  number={22},
  pages={28110-28119},
  doi={10.1109/JSEN.2023.3321841}}

@ARTICLE{R2,
  author={Nan, Fang and Sun, Sihao and Foehn, Philipp and Scaramuzza, Davide},
  journal={IEEE Robotics and Automation Letters}, 
  title={Nonlinear MPC for Quadrotor Fault-Tolerant Control}, 
  year={2022},
  volume={7},
  number={2},
  pages={5047-5054},
  doi={10.1109/LRA.2022.3154033}}

@INPROCEEDINGS{R3,
  author={Han, Xue and Rammal, Rim and Li, Zetao and Cabassud, Michel and Dahhou, Boutaib},
  booktitle={2021 9th International Conference on Systems and Control (ICSC)}, 
  title={Interval Observer-Based Active Fault Tolerant Control for an Intensified Heat Exchanger/Reactor}, 
  year={2021},
  volume={},
  number={},
  pages={133-138},
  doi={10.1109/ICSC50472.2021.9666503}
}

@ARTICLE{R4,
  author={Liu, Yang and Pang, Guochen and Qiu, Jianlong and Chen, Xiangyong and Cao, Jinde},
  journal={IEEE Transactions on Signal and Information Processing over Networks}, 
  title={Fault-Tolerant Control for Output Regulation in Multi-Agent Systems Based on Prescribed-Time Observers}, 
  year={2024},
  volume={10},
  number={},
  pages={729-739},
  doi={10.1109/TSIPN.2024.3458789}}

@ARTICLE{R5,
  author={Doostmohammadian, Mohammadreza and Taghieh, Amin and Zarrabi, Houman},
  journal={IEEE Transactions on Automation Science and Engineering}, 
  title={Distributed Estimation Approach for Tracking a Mobile Target via Formation of UAVs}, 
  year={2022},
  volume={19},
  number={4},
  pages={3765-3776},
  doi={10.1109/TASE.2021.3135834}}

@article{Wafi-DistEstTank,
    author = {Moh Kamalul Wafi and Bambang L. Widjiantoro},
    title = {Distributed Estimation with Decentralized Control for Quadruple-Tank Process},
    journal = {arXiv preprint arXiv:2304.04763},
    year = {2025}, 
    url={https://arxiv.org/abs/2304.04763}, 
}

@article{R8,
title = {Fault-Tolerant Control for Networked Control Systems with Imperfect Measurements*},
journal = {IFAC Proceedings Volumes},
volume = {45},
number = {20},
pages = {1329-1334},
year = {2012},
note = {8th IFAC Symposium on Fault Detection, Supervision and Safety of Technical Processes},
issn = {1474-6670},
doi = {https://doi.org/10.3182/20120829-3-MX-2028.00018},
url = {https://www.sciencedirect.com/science/article/pii/S1474667016349394},
author = {Shun Jiang and Huajing Fang}
}

@INPROCEEDINGS{R9,
  author={Pang, Zhonghua and Zhang, Ji and Zhou, Yuguo and Han, Cunwu},
  booktitle={2017 29th Chinese Control And Decision Conference (CCDC)}, 
  title={Active fault tolerant control of networked systems with sensor fault}, 
  year={2017},
  volume={},
  number={},
  pages={6468-6473},
  doi={10.1109/CCDC.2017.7978337}}

@INPROCEEDINGS{R10,
  author={Cao, Bo and Wu, Yawei and Yao, Lina},
  booktitle={2021 CAA Symposium on Fault Detection, Supervision, and Safety for Technical Processes (SAFEPROCESS)}, 
  title={Fault diagnosis and fault-tolerant control for leader-follower multi-agent systems with time-delay}, 
  year={2021},
  volume={},
  number={},
  pages={1-8},
  doi={10.1109/SAFEPROCESS52771.2021.9693680}}

@INPROCEEDINGS{R11,
  author={Schenk, Kai and Lunze, Jan},
  booktitle={2017 IEEE 56th Annual Conference on Decision and Control (CDC)}, 
  title={Fault-tolerant control in networked systems: A two-layer approach}, 
  year={2017},
  volume={},
  number={},
  pages={6370-6376},
  doi={10.1109/CDC.2017.8264620}}

@INPROCEEDINGS{R12,
  author={Raimondo, Davide M. and Boem, Francesca and Gallo, Alexander and Parisini, Thomas},
  booktitle={2016 IEEE 55th Conference on Decision and Control (CDC)}, 
  title={A decentralized fault-tolerant control scheme based on Active Fault Diagnosis}, 
  year={2016},
  volume={},
  number={},
  pages={2164-2169},
  doi={10.1109/CDC.2016.7798584}}

@ARTICLE{R13,
  author={Gu, Zhou and Shi, Peng and Yue, Dong and Yan, Shen and Xie, Xiangpeng},
  journal={IEEE Transactions on Network Science and Engineering}, 
  title={Fault Estimation and Fault-Tolerant Control for Networked Systems Based on an Adaptive Memory-Based Event-Triggered Mechanism}, 
  year={2021},
  volume={8},
  number={4},
  pages={3233-3241},
  doi={10.1109/TNSE.2021.3107935}}

@ARTICLE{R14,
  author={Lee, Tae H. and Lim, Chee Peng and Nahavandi, Saeid and Roberts, Rodney G.},
  journal={IEEE Systems Journal}, 
  title={Observer-Based $\mathcal{H}_{\infty }$ Fault-Tolerant Control for Linear Systems With Sensor and Actuator Faults}, 
  year={2019},
  volume={13},
  number={2},
  pages={1981-1990},
  doi={10.1109/JSYST.2018.2800710}}

@ARTICLE{R17,
  author={Liu, Xiaoxu and Yuan, Zike and Gao, Zhiwei and Zhang, Wenwei},
  journal={IEEE Transactions on Industrial Informatics}, 
  title={Reinforcement Learning-Based Fault-Tolerant Control for Quadrotor UAVs Under Actuator Fault}, 
  year={2024},
  volume={20},
  number={12},
  pages={13926-13935},
  doi={10.1109/TII.2024.3438241}}

@article{Wafi-Elham,
author = {Javanfar, Elham and Rahmani, Mehdi and Wafi, Moh Kamalul},
title = {Robust Estimation-Based Non-Fragile Control for Discrete-Time Non-Linear Systems},
journal = {International Journal of Robust and Nonlinear Control},
volume = {35},
number = {6},
pages = {2462-2471},
eprint = {https://onlinelibrary.wiley.com/doi/pdf/10.1002/rnc.7806},
year = {2025}
}

@article{WAFI-JRC, 
title={Non-Linear Estimation using the Weighted Average Consensus-Based Unscented Filtering for Various Vehicles Dynamics towards Autonomous Sensorless Design}, 
volume={4}, 
number={1}, 
journal={Journal of Robotics and Control (JRC)}, 
author={Widjiantoro, Bambang L. and Wafi, Moh Kamalul and Indriawati, Katherin}, 
year={2023}, 
month={Mar.}, 
pages={95–107},
url={https://journal.umy.ac.id/index.php/jrc/article/view/16164}, 
DOI={10.18196/jrc.v4i1.16164}, 
abstractNote={The concerns to autonomous vehicles have been becoming more intriguing in coping with the more environmentally dynamics non-linear systems under some constraints and disturbances. These vehicles connect not only to the self-instruments yet to the neighborhoods components, making the diverse interconnected communications which should be handled locally to ease the computation and to fasten the decision. To deal with those interconnected networks, the distributed estimation to reach the untouched states, pursuing sensorless design, is approached, initiated by the construction of the modified pseudo measurement which, due to approximation, led to the weighted average consensus calculation within unscented filtering along with the bounded estimation errors. Moreover, the tested vehicles are also associated to certain robust control scenarios subject to noise and disturbance with some stability analysis to ensure the usage of the proposed estimation algorithm. The numerical instances are presented along with the performances of the control and estimation method. The results affirms the effectiveness of the method with limited error deviation compared to the other centralized and distributed filtering. Beyond these, the further research would be the directed sensorless design and fault-tolerant learning control subject to faults to negate the failures.}, 
}

@INPROCEEDINGS{Wafi-ACC24,
  author={Wafi, Moh Kamalul and Siami, Milad and Sznaier, Mario},
  booktitle={2024 American Control Conference (ACC)}, 
  title={Investigating the Effectiveness of Reinforcement Learning in Closed-Loop Systems with Time Delays}, 
  year={2024},
  volume={},
  number={},
  pages={4149-4154},
  doi={10.23919/ACC60939.2024.10644470}
}

@INPROCEEDINGS{Wafi-CoDIT,
  author={Wafi, Moh Kamalul and Hajian, Rozhin and Shafai, Bahram and Siami, Milad},
  booktitle={2023 9th International Conference on Control, Decision and Information Technologies (CoDIT)}, 
  title={Advancing Fault-Tolerant Learning-Oriented Control for Unmanned Aerial Systems}, 
  year={2023},
  volume={},
  number={},
  pages={1688-1693},
  doi={10.1109/CoDIT58514.2023.10284443}}

@INPROCEEDINGS{Wafi-SIAM,
author = {Moh. Kamalul Wafi and Milad Siami},
title = {A Comparative Analysis of Reinforcement Learning and Adaptive Control Techniques for Linear Uncertain Systems},
booktitle = {2023 Proceedings of the Conference on Control and its Applications (CT)},
chapter = {},
year={2023},
pages = {25-32},
doi = {10.1137/1.9781611977745.4},
eprint = {https://epubs.siam.org/doi/pdf/10.1137/1.9781611977745.4}
}

@ARTICLE{Wafi-AIP,
author = {Wafi, Moh Kamalul},
title = {Filtering module on satellite tracking},
journal = {AIP Conference Proceedings},
volume = {2088},
number = {1},
pages = {020045},
year = {2019},
doi = {10.1063/1.5095297},
URL = {https://aip.scitation.org/doi/abs/10.1063/1.5095297}}

@ARTICLE{Add1,
  author={Fahad, Shah and She, Buxin and Yin, Junjie and Li, Fangxing and Cui, Hantao and Bo, Rui},
  journal={IEEE Transactions on Power Delivery}, 
  title={A Data-Driven Adaptive Control Approach for Enhancing the Dynamic Response of VSGs in Varying Grid Conditions}, 
  year={2025},
  volume={40},
  number={3},
  pages={1421-1433},
  doi={10.1109/TPWRD.2025.3557059}}

@ARTICLE{Add2,
  author={Ruchun, Wen},
  journal={IEEE Access}, 
  title={Adaptive Safe Data Driven Control Strategy for Closed Loop System}, 
  year={2025},
  volume={13},
  number={},
  pages={95876-95887},
  doi={10.1109/ACCESS.2025.3573520}}

@ARTICLE{Add3,
  author={Pinthurat, Watcharakorn and Kongsuk, Prayad and Surinkaew, Tossaporn and Marungsri, Boonruang},
  journal={IEEE Access}, 
  title={An Adaptive Data-Driven-Based Control for Voltage Control Loop of Grid-Forming Converters in Variable Inertia MGs}, 
  year={2024},
  volume={12},
  number={},
  pages={58143-58155},
  doi={10.1109/ACCESS.2024.3392295}}

@ARTICLE{Add4,
  author={Li, Mengshi and Zhang, Huanming and Ji, Tianyao and Wu, Q. H.},
  journal={CSEE Journal of Power and Energy Systems}, 
  title={Fault Identification in Power Network Based on Deep Reinforcement Learning}, 
  year={2022},
  volume={8},
  number={3},
  pages={721-731},
  doi={10.17775/CSEEJPES.2020.04520}}

@ARTICLE{Add5,
  author={Muttaki, Md Rafid and Rahman, Md Habibur and Kulkarni, Akshay and Tehranipoor, Mark and Farahmandi, Farimah},
  journal={IEEE Transactions on Very Large Scale Integration (VLSI) Systems}, 
  title={FTC: A Universal Framework for Fault-Injection Attack Detection and Prevention}, 
  year={2024},
  volume={32},
  number={7},
  pages={1311-1324},
  doi={10.1109/TVLSI.2024.3384531}}

\appendix

\begin{proof}[Proof of Proposition~\ref{prop:e-zero-implies-consensus}]
    If $\hat y_i=y_0, \forall i$ then $\hat{\bar y} = \bar y_0$ and hence $(\mathbb{L}\otimes I_{n_y})\hat{\bar y} = (\mathbb{L}\otimes I_{n_y}) \bar y_0 = (\mathbb{A}_0\otimes I_{n_y})\bar y_0$, implying $\bar e=\mathbf{0}_{mn_y}$ by definition.
    Conversely, assume $\bar e=\mathbf{0}_{mn_y}$ and the condition of Remark~\ref{rem:threshold} hold. We show that this implies $\hat y_i=y_0$ for all $i$. 
    Assume $n_y=1$ and $\bar e=\mathbf{0}_{m}$. Then, $\hat{\bar y} = \mathbb{A}_m\hat{\bar y} + \mathbb{A}_0\bar y_0$. Notice that by \eqref{eq:ComNet:balance}, for each $i\in\{1,\dots,m\}$, we have
    $\sum_{j\in\mathcal{N}_i} w_{ij}+w_{i0}=1$ and $w_{ij},w_{i0}\ge 0$, so
    $\hat y_i$ is a convex combination of $\{\hat y_j : j\in\mathcal{N}_i\}\cup\{y_0\}$.
    
    Let $\tilde{y}_i := \hat y_i - y_0$ and 
    $|\tilde{y}_{i^\star}| = \max_{1\le i\le m} |\tilde{y}_i|$. Subtracting $y_0$ from $\hat y_i = \sum_{j\in\mathcal{N}_i} w_{ij} \hat y_j + w_{i0} y_0$ results in the following $\tilde{y}_i = \sum_{j\in\mathcal{N}_i} w_{ij} \tilde{y}_j$.
    Taking absolute values and using $|\tilde{y}_j|\le |\tilde{y}_{i^\star}|$,
    \begin{equation*}
        |\tilde{y}_{i^\star}|
        \le \sum\nolimits_j w_{i^\star j} |\tilde{y}_j|
        \le |\tilde{y}_{i^\star}|\sum\nolimits_j w_{i^\star j}
        \le |\tilde{y}_{i^\star}|,
    \end{equation*}
    which implies $|\tilde{y}_j| = |\tilde{y}_{i^\star}|$ whenever $w_{i^\star j}>0$ and the signs of $\tilde{y}_j$ and $\tilde{y}_{i^\star}$ coincide. Since at least one node has $w_{i0}>0$ and every node is reachable from $0$ by Remark~\ref{rem:threshold}, we obtain $\tilde{y}_i = 0$ for that node and hence $\tilde{y}_i=0$ for all $i$, i.e., $\hat y_i = y_0$ for all $i$.  In particular $\hat y_i=\hat y_j$ for all $i,j$. For vector outputs $\hat y_i\in\mathbb{R}^{n_y}$, the same argument applies which yields $\hat y_i = y_0$ for all $i$ and
    completes the proof.
\end{proof} 

\begin{proof}[Proof of Theorem~\ref{thm:observer-Hinf}]
Let the Lyapunov function be $V(t) = \bar\epsilon^{\top}(t)\mathbf{P}\bar\epsilon(t)$, with $\mathbf{P}\succ 0$. Using \eqref{eq:FTC:estimation-error}, then
\begin{align*}
    \dot V
    &= \bar\epsilon^{\top}
       \big[\mathbf{P}(\mathbf{F}_1\mathbf{A}_a - \mathbf{L}\mathbf{E}_2)
          + (\mathbf{F}_1\mathbf{A}_a - \mathbf{L}\mathbf{E}_2)^{\top}\mathbf{P}
       \big]\bar\epsilon + 2\,\bar\epsilon^{\top}\mathbf{P}\mathbf{F}_1\mathbf{D}\,\bar v.
\end{align*}
Introduce $\mathbf{H} := \mathbf{P}\mathbf{L}$ to eliminate the bilinear term in $\mathbf{P}$ and $\mathbf{L}$, so that $\mathbf{L} = \mathbf{P}^{-1}\mathbf{H}$. 
By adding and subtracting $\bar\epsilon^{\top}\bar\epsilon$ and $\delta^{2}\bar v^{\top}\bar v$, we obtain
\begin{equation} \label{eq:FTC:thm1-LMI}
    \dot V + \bar\epsilon^{\top}\bar\epsilon - \delta^{2}\bar v^{\top}\bar v
    =
    \begin{bmatrix}
        \bar\epsilon \\[2pt] \bar v
    \end{bmatrix}^{\!\top}
    \begin{bmatrix}
        \Delta & \mathbf{P}\mathbf{F}_1\mathbf{D} \\
        *                     & -\delta^{2}I_{\bar n_v}
    \end{bmatrix}
    \begin{bmatrix}
        \bar\epsilon \\[2pt] \bar v
    \end{bmatrix}.
\end{equation}
where $\Delta = \mathbf{P}\mathbf{F}_1\mathbf{A}_a + \mathbf{A}_a^{\top}\mathbf{F}_1^{\top}\mathbf{P} - \mathbf{H}\mathbf{E}_2 - \mathbf{E}_2^{\top}\mathbf{H}^{\top} + I_{\bar n_x + \bar n_y}$.
Thus, the LMI condition $\Pi<0$ in \eqref{eq:FTC:LMI-observer} is equivalent to
\begin{equation*}
    \dot V + \|\bar\epsilon\|^{2} - \delta^{2}\|\bar v\|^{2} < 0
    \quad\text{for all } (\bar\epsilon,\bar v)\neq 0.
\end{equation*}
Equivalently, $\dot V \le -\|\bar\epsilon\|^{2} + \delta^{2}\|\bar v\|^{2}$. Two conclusions follow immediately: 
(i) for $\bar v(t)\equiv 0$ we achieve $\dot V \le -\|\bar\epsilon\|^{2}\le 0$. Since $V$ is positive definite and radially unbounded in $\bar\epsilon$, this implies that $\bar\epsilon(t)\to 0$ as $t\to\infty$, establishing asymptotic stability; (ii) the inequality $\dot V \le -\|\bar\epsilon\|^{2} + \delta^{2}\|\bar v\|^{2}$ is the standard dissipation inequality with storage function $V$ and supply rate $s(\bar v,\bar\epsilon)
= \delta^{2}\|\bar v\|^{2} - \|\bar\epsilon\|^{2}$. By the bounded–real lemma, this is equivalent to the error system having $\mathcal{H}_{\infty}$ gain (from $\bar v$ to $\bar\epsilon$) strictly less than~$\delta$, completing the proof.
\end{proof}

\begin{proof}[Proof of Theorem~\ref{thm:network-Hinf}]
Let $\mathbf{Q}\succ 0$ be the Lyapunov matrix and consider the quadratic storage function
$V(\bar x) = \bar x^\top\mathbf{Q}\bar x$.
Let the signal $\bar \vartheta\coloneqq\mathbf{K}\mathbf{E}_1 \bar \epsilon$ enter the performance inequality with weight
$\alpha>0$. The derivative of $V$ along \eqref{eq:FTC:closed-loop} is 
\begin{equation*}
    \dot V
    =
    \bar x^\top
    \big(
        \mathbf{Q}(\mathbf{A}+\mathbf{B}\mathbf{K})
      + (\mathbf{A}+\mathbf{B}\mathbf{K})^\top\mathbf{Q}
    \big)\bar x
    + 2\bar x^\top\mathbf{Q}\mathbf{B}_\theta\bar\theta.
\end{equation*}
To enforce the dissipation inequality
\begin{equation*}
    \dot V + \bar x^\top\bar x
        - \alpha\,\bar \vartheta^\top\bar \vartheta
        - \delta^2 \bar v^\top\bar v < 0, \quad \forall(\bar x,\bar \vartheta,\bar v)\neq 0
\end{equation*}
define $\hat\Delta \coloneqq \mathbf{Q}(\mathbf{A}+\mathbf{B}\mathbf{K}) + (\mathbf{A}+\mathbf{B}\mathbf{K})^\top\mathbf{Q} + I_{\bar n_x}$. Then the above inequality is equivalent to
\begin{equation}
    \label{eq:LMI-QK-appendix}
    \begin{bmatrix}
        \bar x \\[1mm] \bar \vartheta \\[1mm] \bar v
    \end{bmatrix}^\top
    \begin{bmatrix}
        \hat\Delta     & -\mathbf{Q}\mathbf{B} & \mathbf{Q}\mathbf{D} \\
        *              & -\alpha I_{\bar n_u}            & 0          \\
        *              & *                    & -\delta^2 I_{\bar n_v}
    \end{bmatrix}
    \begin{bmatrix}
        \bar x \\[1mm] \bar \vartheta \\[1mm] \bar v
    \end{bmatrix}
    < 0,
\end{equation}
for all nonzero $(\bar x,\bar \vartheta,\bar v)$.
Integrating the dissipation inequality yields $\|\bar x\|_2^2 < \alpha\|\bar\vartheta\|_2^2 + \delta^2\|\bar v\|_2^2,$
which establishes the weighted $\mathcal{H}_\infty$ performance.
By the bounded real lemma, \eqref{eq:LMI-QK-appendix} implies that the closed--loop
system is asymptotically stable. The bound $\|\bar x\|_{2} < \gamma\|\bar v\|_{2}$
then follows from Remark~\ref{rem:gamma-delta}.

To obtain an LMI that is affine in the decision variables, introduce the
change of variables $\mathbf{R} \coloneqq \mathbf{Q}^{-1}$ and $\mathbf{G} \coloneqq \mathbf{K}\mathbf{R}$.
Pre- and post--multiplying \eqref{eq:LMI-QK-appendix} by
$\diag\{\mathbf{R}, I_{\bar n_u}, I_{\bar n_v}\}$ and its transpose yields the equivalent
matrix inequality
\begin{equation*}
    \begin{bmatrix}
        \tilde\Delta & -\mathbf{B} & \mathbf{D}    \\
        *            & -\alpha I_{\bar n_u}   & 0  \\
        *            & *          & -\delta^{2} I_{\bar n_v} 
    \end{bmatrix} < 0,
\end{equation*}
where $\tilde\Delta = \mathbf{A}\mathbf{R} + \mathbf{R}\mathbf{A}^\top + \mathbf{B}\mathbf{G} + \mathbf{G}^\top\mathbf{B}^\top + \mathbf{R}\mathbf{R}$.

The term $\mathbf{R}\mathbf{R}$ is quadratic in the decision variable
$\mathbf{R}$.  To obtain an LMI that is affine in $(\mathbf{R},\mathbf{G})$ we
apply a Schur complement. Writing $\tilde\Delta = \bar\Delta + \mathbf{R}\mathbf{R}$ where $\bar\Delta \coloneqq \mathbf{A}\mathbf{R} + \mathbf{R}\mathbf{A}^\top + \mathbf{B}\mathbf{G} + \mathbf{G}^\top\mathbf{B}^\top$,
and using the equivalence
\begin{equation*}
    \bar\Delta + \mathbf{R}\mathbf{R} < 0
    \;\Longleftrightarrow\;
    \begin{bmatrix}
        \bar\Delta & \mathbf{R} \\
        *          & -I_{\bar n_x} 
    \end{bmatrix} < 0,
\end{equation*}
we arrive at the LMI \eqref{eq:FTC:LMI-control-4x4}, which is affine in $(\mathbf{R},\mathbf{G})$ and thus convex. Therefore, any pair $(\mathbf{R},\mathbf{G})$ satisfying \eqref{eq:FTC:LMI-control-4x4} defines, via $\mathbf{K}=\mathbf{G}\mathbf{R}^{-1}$, a feedback law that stabilizes the closed--loop system and achieves the desired weighted $\mathcal{H}_\infty$ bound.
\end{proof}

\begin{proof}[Proof of Theorem~\ref{thm:ISS-error}]
Here $\mathbb{L}$ is the augmented Laplacian and is invertible under Remark~\ref{rem:threshold} and $\bar x_0$ is constant. Differentiating \eqref{eq:FTC:coop-error-output} and using
$\bar x = (\mathbb{L}\otimes I_{n_x})^{-1} \big[\bar e_x + (\mathbb{A}_0\otimes I_{n_x})\bar x_0 \big]$ (obtained from rearrangement of \eqref{eq:FTC:coop-error-output}) gives
\begin{equation} \label{eq:FTC:thm3:error-dynamics}
    \begin{aligned}
    \dot{\bar e}_x &=
    (\mathbb{L}\otimes I_{n_x})(\mathbf{A}+\mathbf{B}\mathbf{K})\bar x
    + (\mathbb{L}\otimes I_{n_x})\mathbf{B}_\theta \bar\theta^\ast \\
    &= \Phi \bar e_x + \bar \phi_0 + \mathbf{B}_{\phi}\bar\theta^\ast,
    \end{aligned}
\end{equation}
where $\Phi \coloneqq (\mathbb{L}\otimes I_{n_x})(\mathbf{A}+\mathbf{B}\mathbf{K})(\mathbb{L}\otimes I_{n_x})^{-1}$, $\bar \phi_0 \coloneqq (\mathbb{L}\otimes I_{n_x})(\mathbf{A}+\mathbf{B}\mathbf{K})(\mathbb{L}\otimes I_{n_x})^{-1}(\mathbb{A}_0\otimes I_{n_x})\bar x_0$, and $\mathbf{B}_\phi \coloneqq (\mathbb{L}\otimes I_{n_x})\mathbf{B}_\theta$. Since Theorem~\ref{thm:network-Hinf} guarantees that
$\mathbf{A}+\mathbf{B}\mathbf{K}$ is Hurwitz and $\mathbb{L}\otimes I_{n_x}$ is
nonsingular, the matrix $\Phi$ is also Hurwitz (similarity transformation).

Because $\Phi$ is Hurwitz, for any chosen positive definite matrix
$\mathbf{Q}\in\mathbb{R}^{\bar n_x\times\bar n_x}$ there exists a unique
$\mathbf{P}_e\succ 0$ solving the Lyapunov equation
$\Phi^{\top}\mathbf{P}_e + \mathbf{P}_e\Phi = -\mathbf{Q}.$
Let $\bar e^\star$ be any equilibrium of \eqref{eq:FTC:thm3:error-dynamics} in the
disturbance--free case $\bar\theta^\ast\equiv 0$, i.e. a solution of $\Phi\bar e^\star + \bar \phi_0 = 0$. 
Define the shifted error variable $\tilde{\bar e} \coloneqq \bar e_x - \bar e^\star$, which measures deviation from the equilibrium.
Then $\tilde{\bar e}$ satisfies
\begin{equation}\label{eq:FTC:thm3:e-tilde}
    \dot{\tilde{\bar e}} =
    \Phi \tilde{\bar e} + \mathbf{B}_\phi\bar\theta^\ast.
\end{equation}

Consider the Lyapunov function $V(\tilde{\bar e}) = \tilde{\bar e}^{\top}\mathbf{P}_e\tilde{\bar e}$ with $\mathbf{P}_e\succ 0$.
Differentiating along trajectories of \eqref{eq:FTC:thm3:e-tilde} gives
\begin{align*}
    \dot V
    &= 
    \tilde{\bar e}^{\top}(\Phi^{\top}\mathbf{P}_e + \mathbf{P}_e\Phi)\tilde{\bar e}
    + 2\tilde{\bar e}^{\top}\mathbf{P}_e \mathbf{B}_\phi\bar\theta^\ast
    = -\tilde{\bar e}^{\top}\mathbf{Q}\tilde{\bar e}
       + 2\tilde{\bar e}^{\top}\mathbf{P}_e \mathbf{B}_\phi\bar\theta^\ast.
\end{align*}
Let $\lambda_{\min}(\mathbf{Q})$ denote the smallest eigenvalue of $\mathbf{Q}$ and define
$\kappa := \|\mathbf{P}_e \mathbf{B}_\phi\|$. Then, $\dot V \le -\lambda_{\min}(\mathbf{Q})\|\tilde{\bar e}\|^2 + 2\kappa \|\tilde{\bar e}\|\|\bar\theta^\ast\|$.
Next, applying the Young inequality $2ab\le c a^2 + c^{-1} b^2$ to $2\kappa \|\tilde{\bar e}\|\|\bar\theta^\ast\|$ with $a = \|\tilde{\bar e}\|$, $b = \kappa \|\bar\theta^\ast\|$, and $c = \lambda_{\min}(\mathbf{Q})/2$, yields
\begin{align*}
    \dot V &\le
    -\frac{\lambda_{\min}(\mathbf{Q})}{2}\|\tilde{\bar e}\|^2
    + \frac{2\kappa^2}{\lambda_{\min}(\mathbf{Q})}\|\bar\theta^\ast\|^2.
\end{align*}
Using the bounds $\lambda_{\min}(\mathbf{P}_e)\|\tilde{\bar e}\|^2 \le V(\tilde{\bar e}) \le \lambda_{\max}(\mathbf{P}_e)\|\tilde{\bar e}\|^2,$
then $\|\tilde{\bar e}\|^2\ge V(\tilde{\bar e})/\lambda_{\max}(\mathbf{P}_e)$. Hence $\dot V \le -\alpha V + \beta \|\bar\theta^\ast\|^2,$
with $\alpha \coloneqq \lambda_{\min}(\mathbf{Q})/[2\lambda_{\max}(\mathbf{P}_e)]$ and $\beta \coloneqq 2\kappa^2/\lambda_{\min}(\mathbf{Q})$.

The linear differential inequalities imply
the ISS estimate
\begin{equation*}
    \|\tilde{\bar e}(t)\|
    \le
    c_1 e^{-c_2 t}\|\tilde{\bar e}(0)\|
    + c_3 \sup_{0\le \tau\le t}\|\bar\theta^\ast(\tau)\|,
\end{equation*}
where the constants are $c_1 = \sqrt{\lambda_{\max}(\mathbf{P}_e)/\lambda_{\min}(\mathbf{P}_e)}$, $c_2 = \alpha/2$, and $c_3 = \sqrt{\beta/\alpha\lambda_{\min}(\mathbf{P}_e)}$.
Since $\bar e_x = \tilde{\bar e} + \bar e^\star$ and
$\bar e^\star$ is constant, this yields a bound of the form
\eqref{eq:ISS-bound} with new constants $c_1,c_2,c_3>0$.

In the disturbance--free case $\bar\theta^\ast\equiv 0$, we have
$\dot V \le -\alpha V$ and therefore $\tilde{\bar e}(t)\to 0$ exponentially,
so $\bar e_x(t)\to \bar e^\star$. For constant $\bar x_0$, the integral action ensures that $\bar e^\star = 0$, yielding $\bar e_x(t)\to 0$. 
Finally, Proposition~\ref{prop:e-zero-implies-consensus} implies that
the network achieves cooperative consensus tracking of the setpoint.
\end{proof}

\begin{proof}[Proof of Corollary~\ref{cor:output-tracking}]
From $\bar y = \mathbf{C}\bar x$ with $\mathbf{C}=\diag\{C_1,\dots,C_m\}$, we obtain
\begin{equation*}
    \bar e_y
    = (\mathbb{L}\otimes I_{n_y})\mathbf{C}\bar x
    - (\mathbb{A}_0\otimes I_{n_y})\bar y_0.
\end{equation*}
Since $\mathbf{C}$ is a bounded linear operator, there exists a constant $\kappa_C>0$ such that
$\|\bar e_y\| \le \kappa_C \|\bar e_x\|.$
The result then follows directly from the ISS bound in Theorem~\ref{thm:ISS-error}.
\end{proof}

\end{document}